\documentclass[12pt]{article}
\usepackage{ct}

\specs{n}{vol (issue)}{year}{}

\dateline{Mar 17, 2023}{Mar 25, 2025}{TBD}

 % TODO: check that the keywords are good


\title{Bounds on Unique-Neighbor Codes}

\author[1]{Nathan Linial\thanks{Supported in part by an ERC Grant 101141253, ``Packing in Discrete Domains -- Geometry and Analysis'' and a NSF-BSF research grant ``Global Geometry of Graphs''.}}
\author[2]{Edan Orzech\thanks{Work was done while at the Hebrew University.}}

\affil[1]{%
School of Computer Science and Engineering, Hebrew University, Jerusalem 91904, Israel.

\email{nati@cs.huji.ac.il}%
}

\affil[2]{%
Department of Electrical Engineering and Computer Science, MIT CSAIL, Cambridge MA 02139, USA.

\email{iorzech@csail.mit.edu}%
}

\usepackage{mathtools}

\newcommand{\F}{\mathbb{F}}
\newcommand{\e}{\varepsilon}
\newcommand\numberthis{\addtocounter{equation}{1}\tag{\theequation}}

\DeclareMathOperator{\lcm}{lcm}
\DeclarePairedDelimiter{\floorfixed}{\lfloor}{\rfloor}
\newcommand{\floor}{\floorfixed*}
\DeclarePairedDelimiter{\ceilfixed}{\lceil}{\rceil}
\newcommand{\ceil}{\ceilfixed*}

\begin{document}
%\writedatatofile

\maketitle

% ABSTRACT
% CT papers must include an abstract. The abstract should consist of a
% succinct statement of background followed by a listing of the
% principal new results that are to be found in the paper. The abstract
% should be informative, clear, and as complete as possible. Phrases
% like "we investigate..." or "we study..." should be kept to a minimum
% in favor of "we prove that..."  or "we show that...".  Do not
% include equation numbers, unexpanded citations (such as "[23]"), or
% any other references to things in the paper that are not defined in
% the abstract. The abstract may be distributed without the rest of the
% paper so it must be entirely self-contained.  Try to include all words
% and phrases that someone might search for when looking for your paper.
% You can use some basic LaTeX commands in the abstract, but not any
% user defined macros. 

\begin{abstract}
    Recall that a \emph{binary linear code of length $n$} is a linear subspace $\mathcal{C} = \{x\in \mathbb{F}_2^n \mid Ax=0\}$.
    Here the \emph{parity check matrix} $A$ is a binary $m\times n$ matrix of rank $m$.
    We say that $\mathcal{C}$ has \emph{rate} $R=1-\frac mn$.
    Its \emph{distance}, denoted $\delta n$ is the smallest Hamming weight of a non-zero vector in $\mathcal{C}$.
    The \emph{rate vs.\ distance problem} for binary linear codes is a fundamental open problem in coding theory, and a fascinating question in discrete mathematics.
    It concerns the function $R_L(\delta)$, the largest possible rate $R$ for given $0\le\delta\le 1$ and arbitrarily large length $n$.
    Here we investigate a variation of this fundamental question that we describe next.
      
    Clearly, $\mathcal{C}$ has distance $\delta n$, if and only if for every $0<n'<\delta n$, every $m\times n'$ submatrix of $A$ has a row of odd weight.
    Motivated by several problems from coding theory, we say that $A$ has the \emph{unique-neighbor} property with parameter $\delta n$, if every such submatrix has a row \emph{of weight $1$}.
    Let $R_U(\delta)$ be the largest possible asymptotic rate of linear codes with a parity check matrix that has this stronger property.
    Clearly, $R_U(\cdot), R_L(\cdot)$ are non-increasing functions, and $R_U(\delta)\le R_L(\delta)$ for all $\delta$.
    Also, $R_U(0) = R_L(0) = 1$, and $R_U(1) = R_L(1) = 0$, so let $0\le\delta_U \le\delta_L\le 1$ be the smallest values of $\delta$ at which $R_U$ resp.\ $R_L$ vanish.
    It is well known that $\delta_L=\frac{1}{2}$ and we conjecture that $\delta_U$ is strictly smaller than~$\frac{1}{2}$, i.e., the rate of linear codes with the unique-neighbor property is more strictly bounded.
    While the conjecture remains open, we prove here several results supporting it.
      
    The reader is not assumed to have any specific background in coding theory, but we occasionally point out some relevant facts from that area.
\end{abstract}

% TABLE OF CONTENTS, LIST OF FIGURES, LIST OF TABLES
% Please, do not include a table of contents, a list of figures, or a
% list of tables. They will be removed by the editors (and the command
% is actually redefined in the ct.sty file).

%%%%%%%%%%%%%%%%%%%%%%%%%%%%%%%%%%%%%%%%%%%%%%%%%%%
%%%%%%%%%%%%%%%%%%%%%%%%%%%%%%%%%%%%%%%%%%%%%%%%%%%

\section{Introduction and Main Problems}

We consider here only binary codes $\mathcal{C}\subseteq \{0,1\}^n$ of length $n$.
As usual, we denote the \emph{rate} of $\mathcal{C}$ by $R=R(\mathcal{C})=\frac 1n\log_2|\mathcal{C}|$ and its distance by $\delta n=\text{dist}(\mathcal{C})=\min_{x\neq y,\,x,y\in \mathcal{C}} d_H(x,y)$, where $d_H$ stands for the Hamming distance.
A fundamental open problem in coding theory seeks the best possible tradeoff between $R$ and $0\le\delta\le1$.
We refer to this as

\begin{problem}\label{R_of_delta}
    Determine, or estimate the real function
    \begin{align*}
        R(\delta)=\limsup_{n\to\infty}\{R(\mathcal{C})\mid\mathcal{C}\subseteq \{0,1\}^n,~\text{dist}(\mathcal{C})\ge\delta n\}.
    \end{align*}
\end{problem}

A \emph{linear code} is a linear subspace of the vector space $\mathbb{F}_2^n$, which we identify with $\{0,1\}^n$.
Such a code can be defined in terms of a \emph{parity check matrix} $A$ which is a $\ceil{(1-R)n} \times n$ binary matrix.
Namely, $\mathcal{C}=\{x\in\F_2^n\mid Ax=0\}$.

Let $S$ be a nonempty set of columns in a binary matrix and let $z$ be the sum \emph{over the integers} of the columns in $S$.
We say that the set $S$ is \emph{$1$-free} if no entry of $z$ equals $1$, and we call $S$ \emph{even}, if all entries of $z$ are even integers.
The number of $1$-entries in a binary vector is said to be its \emph{weight} or \emph{sum}.

Here are some basic comments:

\begin{remark}
    \begin{itemize}
        \item
        A given linear code $\mathcal{C}$ can clearly have various distinct parity
        check matrices~$A$.
        While the distance of $\mathcal{C}$ does not depend on the choice of $A$,
        different parity check matrices of the same code may have different $1$-free sets.
        \item 
        If we delete from $A$ any row that linearly depends on the other rows, the code 
        $\mathcal{C}$ does not change.
        Consequently, all parity check matrices herein can, with no loss in generality be assumed to have full row rank.
        \item 
        Coding theorists refer to a $1$-free set as a \emph{stopping set}, see \cref{sec:context}.
        \item 
        We consider below also vanishingly small rates $R=o_n(1)$, in which case we may speak of an $n\times(n+f(n))$ parity check matrix with $f(n)=o(n)$.
        Although this deviates from our standard of considering length-$n$ codes, the asymptotic conclusions remain unchanged.
    \end{itemize}
\end{remark}

We are ready now to introduce the protagonists of this article.

\begin{definition}
    Let $A$ be a binary $m\times n$ matrix of rank $m$.
    Define:
    \begin{itemize}
        \item 
        The smallest cardinality of a nonempty even set of columns in $A$, i.e., the \emph{distance} of the binary linear code with parity check $A$ is denoted by $\e(A)$.
        \item
        The smallest cardinality of a nonempty $1$-free set of columns in $A$ is denoted by $u(A)$.
        \item
        The maximum value of $\e(A)$ over all binary $m\times n$ matrices is denoted $\e(m,n)$.
        \item
        The maximum value of $u(A)$ over all binary $m\times n$ matrices is denoted $u(m,n)$.
    \end{itemize}
\end{definition}

The best available estimates of $\e(m,n)$ are extensively tabulated, e.g., in \cite{grassl2007bounds}.
We seek to likewise determine or estimate the values of $u(m,n)$.
Note the following:

\begin{remark}
    \begin{itemize}
        \item
        The claim that $T\le u(m,n)$ means that there exists a binary $m\times n$ matrix~$A$ of rank $m$ where every nonempty set of at most $T$ columns has a row of weight $1$.
        \item 
        Proving that $u(m,n)\le T$ entails showing that in every binary $m\times n$ matrix $A$ of rank $m$ there is a set of $T$ or fewer columns for which no row weighs $1$.
    \end{itemize}
\end{remark}

In analogy with \cref{R_of_delta}, the following asymptotic question suggests itself:

\begin{problem}\label{L_of_delta}
    Determine, or estimate the real function
    \begin{align*}
        R_L(\delta)\coloneqq\limsup_{n\to\infty}\{R\mid\text{there exists a~} \ceil{(1-R)n} \times n \text{~binary matrix $A$ with~}\e(A)\ge\floor{\delta n}\}.
    \end{align*}
    In other words, $R_L(\delta)$ is the smallest real $R$ such that 
    \begin{align*}
        \begin{array}{c}
            \text{For any~} \rho > R \text{~and large enough~}n, \text{~every~} \ceil{(1-\rho)n} \times n \text{~binary matrix}\\[1em]
            \text{has an even set of~}\le\delta n\text{~columns.}
        \end{array}
    \end{align*}
\end{problem}

And for $1$-free sets:

\begin{problem}\label{U_of_delta}
    Determine, or estimate the real function
    \begin{align*}
        R_U(\delta)\coloneqq\limsup_{n\to\infty}\{R\mid\text{there exists a~} \ceil{(1-R)n} \times n \text{~binary matrix with~}u(A)\ge\floor{\delta n}\}.
    \end{align*}
    In other words, $R_U(\delta)$ is the smallest real $R$ such that 
    \begin{align*}
        \begin{array}{c}
            \text{For any~} \rho > R \text{~and large enough~}n, \text{~every~} \ceil{(1-\rho)n} \times n \text{~binary matrix}\\[1em]
            \text{has a $1$-free set of~}\le\delta n\text{~columns.}
        \end{array}
    \end{align*}
\end{problem}

Clearly,
\begin{align*}
    R_U(\delta)\le R_L(\delta)\le R(\delta)\text{~for all~}0\le\delta.
\end{align*}
At present, we cannot even rule out the possibility that all these three functions are, in fact, identical.
It is easily verified that (i) All three are non-increasing functions of $\delta$, and (ii) $R(0)=R_L(0)=R_U(0)=1$.
It is also well known that $R(\delta), R_L(\delta)$ are positive for $\delta<\frac12$ (by the Gilbert-Varshamov bound for linear codes~\cite{gilbert1952comparison,varshamov1957estimate}) and $R(\delta), R_L(\delta)=0$ for $\frac12\le\delta$ (e.g., by the Plotkin bound~\cite{plotkin1960binary}).

We believe that the strict inequality $R_U(\delta)<R_L(\delta)$ holds for at least some of the range 
$0<\delta<\frac12$.
More specifically that $R_U$ vanishes already at some $\delta_0<\frac12$.
Concretely, we state

\begin{conjecture}\label{main:conj}
    There exists some $\epsilon_0>0$ such that function $R_U(\delta)$ vanishes already at 
    $\delta=\frac 12 -\epsilon_0$.
    
    In other words, for every $R>0$ and large enough $n$, every $\ceil{(1-R)n} \times n$ binary matrix has a $1$-free set of at most $(\frac 12 -\epsilon_0)n$ columns.
\end{conjecture}

It is conceivable that the same conclusion holds even for $R=o_n(1)$ and with a specified $o(1)$-term.

\begin{conjecture}\label{conj:little_oh}
    There exists some $\epsilon_0>0$ and a positive function $\eta(n)=o(n)$ such that for every large enough $n$, every $\ceil{n-\eta(n)} \times n$ binary matrix has a $1$-free set of at most $(\frac 12 -\epsilon_0)n$ columns.
\end{conjecture}

Let $A$ be a parity check matrix of a linear code $\mathcal{C}\subseteq \{0,1\}^m$.
Of course $\mathcal{C}$ remains invariant under elementary row operations on $A$.
Also, distances among vectors in $\mathcal{C}$ remain unchanged as $A$'s columns get permuted.
Consequently, in the study of $R_L(\delta)$ as in \cref{L_of_delta}, there is no loss of generality in assuming that $A$ is in \emph{standard form}, i.e., its first $n$ columns form an order-$n$ identity matrix.
We pose:

\begin{problem}\label{with_identity_minor}
    Let $n, k$ be positive integers, and let $A$ be a binary $n\times(n+k)$ matrix whose first $n$ columns form the order-$n$ identity matrix.
    How large can $u(A)$ be?
\end{problem}

Below, we often use the invariance of $u$ and $\e$ under row and column permutations.

\subsection{Unique-neighbor codes in other contexts}\label{sec:context}

The function $u$ was previously (e.g., \cite{di2002finite, kashyap2003stopping, orlitsky2005stopping, schwartz2006stopping}) defined under the name of the \emph{stopping distance} of a matrix.
This notion arises in the study of binary erasure channels (see, e.g., \cite{guruswami2006iterative}).
For further algorithmic aspects of stopping sets see \cite{jiang2010stopping, jiang2011stopping, luby2001efficient, price2017survey, rathi2006asymptotic}.

The closely related notion of \emph{unique-neighbor expansion} of codes appears in the study of message-passing algorithms and expander codes (see \cite{hoory2006expander} for a survey).
Such algorithms offer an approach to the \emph{decoding} of linear codes $\mathcal{C}=\{x\in\F_2^m\mid Ax=0\}$.
Here $A$ is viewed as the bipartite adjacency matrix $(U,V;E)$ of the code's \emph{factor graph} (or Tanner graph~\cite{tanner1981recursive}).
Here $U, V$ are the sets of $A$'s rows and columns, and edges correspond to $1$-entries in $A$.
Message-passing algorithms such as \emph{belief propagation}~\cite{gallager1962low} work by iteratively passing messages between vertices in $U$ and those in $V$.

When the factor graph is a bounded-degree expander graph, we say that $\mathcal{C}$ is an \emph{expander code}.
Such codes belong to the class of low density parity check (LDPC) codes introduced by Gallager \cite{gallager1962low}.
An important feature of such codes is that they can be efficiently decoded, using message-passing algorithms~\cite{gallager1962low, richardson2001capacity, zyablov1975estimation}.
In a highly influential paper \cite{sipser1996expander}, Sipser and Spielman showed that message-passing algorithms can efficiently decode expander codes even when linearly many (in $n$) errors occur.
The performance of the algorithm depends on the unique-neighbor expansion of the graph's bipartite adjacency matrix.
For more on this subject, see \cite{dowling2017fast, guruswami2006iterative, hoory2006expander, richardson2008modern, shokrollahi2004ldpc, viderman2013linear}.

Unique-neighbor expanders and unique-neighbor codes are still not sufficiently well understood.
Alon and Capalbo \cite{alon2002explicit} found explicit constructions of bipartite graphs which are $(\alpha, \beta)$-unique-neighbor expander graphs, where $\alpha,\beta$ are some positive absolute constants and $\frac{|V|}{|U|}$ is bounded away from $1$.
See also \cite{becker2016symmetric} and~\cite{ben2009tensor} for more.
In a recent paper~\cite{hsieh2024explicit}, explicit unique-neighbor expanders are constructed with $\alpha$ bounded away from~$0$ and $\beta\approx\frac{3}{5}$, which surpasses spectral methods that achieve $\beta\le\frac{1}{2}$.

%%%%%%%%%%%%%%%%%%%%%%%

\section{Our New Results}

Our work addresses Problems \ref{U_of_delta} and \ref{with_identity_minor}.
Problems \ref{R_of_delta} and \ref{L_of_delta} are mentioned here for context only.

\begin{enumerate}
    \item\label{pt1}
    \cref{thm:min-weight rows} shows that matrices $A$ in which all row sums are $9$ or more, satisfy \cref{main:conj} with $\epsilon_0=0.03$ even when $R=0$.
    \item\label{pt2}
    \cref{thm:log_gap} shows that the conclusion of \cref{conj:little_oh} fails when $\eta(n)<\log_2(n)$.
    \item
    \cref{thm:u_I bound} answers \cref{with_identity_minor}: 
    $u_I(n,n+k)=\frac{n}{H_k}\pm O_n(1)$.
    Here $H_k$ is the $k$-th harmonic sum.
    In particular we prove \cref{main:conj} for matrices in standard form.
    \item
    Clearly $u(m,n)\le\e(m,n)$ for all $m$ and $n$.
    \cref{thm:sec5 thm} shows that $u=\e$ when $n-m\le3$ and that $(m,n)=(4,8)$ is the first case where $u<\e$.
\end{enumerate}

%%%%%%%%%%%%%%%%%%%%%%%

\section{With Lower Bounds on Row Weights}

As we show next, matrices of sufficiently large row weights satisfy \cref{main:conj}.
In contrast, small $1$-free sets seem harder to find in sparse matrices.
The case of row weights $3$ may be the hardest.
The theorem below, as well as \cref{thm:log_gap} and \cref{thm:sec5 thm} exhibit matrices with row weights all either $3$ or $4$ for which the conclusion of \cref{main:conj} fails to hold.

\begin{theorem}\label{thm:min-weight rows}
    \begin{enumerate}
        \item\label{ge9}
        If every row in a binary $n\times n$ matrix $A$ has weight at least $9$, then $u(A)<0.47n$.
        Namely, $A$ must have an $n\times n'$ submatrix with $n'<0.47n$ in which no row has weight $1$.
        \item\label{le4}
        On the other hand, for every $n$ there exists a binary $n\times n$ matrix $A$ where every row has weight~$4$, such that $u(A)=\ceil{\frac{n}{2}}$.
        Namely, every $n\times n'$ submatrix of $A$ with no row of weight $1$ satisfies $n'\ge\frac{n}{2}$.
        \item\label{wt3}
        For every $n$ there also exists a binary $n\times n$ matrix $B$ where every row has weight $3$, such that ${u(B)=\ceil{\frac{2n}{3}}}$.
    \end{enumerate}
\end{theorem}

\begin{proof}
    \cref{ge9}:
    Let $Z^{(1)}$ be a random set of columns in $A$ that is obtained by picking every column independently with probability $\rho$.
    Let $B^{(1)}$ be the resulting $[n]\times Z^{(1)}$ submatrix of $A$, and let $X_0^{(1)}, X_1^{(1)}$ the set of rows in $B^{(1)}$ of weight zero, resp.\ one.
    For every row $i\in X_1^{(1)}$ add some column to $Z^{(1)}$ whose $i$-th entry is $1$, and let $Z^{(2)}\supseteq Z^{(1)}$ be the resulting, extended column set.
    We denote the $[n]\times Z^{(2)}$ submatrix of $A$ by $B^{(2)}$.
    Let $X_0^{(2)}, X_1^{(2)}$ be the set of rows in $B^{(2)}$ of weight $0$, resp.\ $1$.
    Note that $X_0^{(2)}, X_1^{(2)}\subseteq X_0^{(1)}$.
    We proceed in the same way with $X_0^{(2)}, X_1^{(2)}$ until no row in the submatrix has weight $1$, and let $Z=\cup_k Z^{(k)}$ be the column set of the resulting matrix.
    Note that
    \begin{align*}
        |Z|\le |Z^{(1)}|+|X_0^{(1)}|+|X_1^{(1)}|,
    \end{align*}
    because every column in $Z\setminus Z^{(1)}$ accounts for some row in $X_0^{(1)}\cup X_1^{(1)}$ that no other column in $Z\setminus Z^{(1)}$ does.
    Therefore
    \begin{align*}
        \mathbb{E}(|Z|)\le \mathbb{E}(|Z^{(1)}|+|X_0^{(1)}|+|X_1^{(1)}|).
    \end{align*}
    We have
    \begin{align*}
        \mathbb{E}(|Z^{(1)}|)=\rho n\,;~~~\mathbb{E}(|X_0^{(1)}|)\le n(1-\rho)^c+o(n)\,;~~~\mathbb{E}(|X_1^{(1)}|)\le nc\rho(1-\rho)^{c-1}+o(n),
    \end{align*}
    if the Hamming weight of every row in $A$ is at least $c$.
    The claim follows by observing that 
    \begin{align*}
        \mathbb{E}(|Z|)<0.47n \text{~for~} c=9,\,\rho=0.3757.
    \end{align*}
    
    \cref{le4}:
    Let $A$ be the $n\times n$ binary matrix whose rows are comprised of all $n$ cyclic rotations of the vector $1^40^{n-4}$.
    Given a vector $x\in\{0,1\}^n\setminus\{\mathbf{0}\}$ (here $\mathbf{0}$ is the all-zero vector of length~$n$), let the vector $Ax$ be defined by integer arithmetic.
    Clearly, $Ax\in\{0,1,2,3,4\}^n$.
    Multiply on the left by the all-$1$ vector to conclude that $\|Ax\|_1 = 4\|x\|_1$.
    Note next that if $(Ax)_i=0$, then $(Ax)_{i+1 \bmod n}$ is either $0$ or $1$.
    Therefore, if $Ax$ has no $1$ coordinates, then all its coordinates are at least $2$, so that $4\|x\|_1=\|Ax\|_1 \ge 2n$ and $\|x\|_1 \ge\frac{n}{2}$, so $u(A)\ge\frac{n}{2}$.
    A $1$-free set of size~$\ceil{\frac{n}{2}}$ is the set $\{1,3,\ldots,2\ceil{\frac{n}{2}}-1\}$.
    Let $x\in\{0,1\}^n$ be its indicator vector.
    Since every two neighboring columns in the selected set have distance at most $2$ (where the distance between $1$ and $n$ is $1$), then for every row $i$, $\|(Ax)_i\|_1\ge2$ as claimed.
    
    \cref{wt3}:
    The matrix $B$ is similar to $A$ but with row weights of $3$, and a similar argument yields $u(B)=\ceil{\frac{2n}{3}}$.
    To show that $u(B)\ge\frac{2n}{3}$, let $x\in\{0,1\}^n$ be the indicator vector of a $1$-free set of columns in $B$.
    As before, $3\|x\|_1=\|Bx\|_1\ge2n$, so that $\|x\|_1\ge\frac{2n}{3}$ as claimed.
    To show that $u(B)\le \ceil{\frac{2n}{3}}$, notice that the set $\{j\in[n]\mid j\not\equiv0\bmod3\}$ is $1$-free column set with $\ceil{\frac{2n}{3}}$ columns.
    The conclusion follows.
\end{proof}

\cref{le4} of \cref{thm:min-weight rows} reflects on the validity of \cref{main:conj}.
It shows that to guarantee the existence of small $1$-free sets of columns, we must consider matrices with more columns than rows.
This statement is made quantitative in \cref{thm:log_gap}.

We suspect that \cref{ge9} of \cref{thm:min-weight rows} remains valid even when all row weights are at least~$5$.
However, this seems to require a substantial new idea.

%%%%%%%%%%%%%%%%%%%%%%%

\section{Matrices in Standard Form}

We denote by $u_I(m,n)$ the maximum of $u(A)$ for a binary $m\times n$ matrix in standard form $A=[I_m | B]$, where $I_m$ stands for the identity $m\times m$ matrix.
Answering \cref{with_identity_minor}, we give an upper bound on $u_I(m,n)$ that is tight in infinitely many cases.

\begin{theorem}\label{thm:u_I bound}
    For every positive integer $k$ and $n\to\infty$, every binary $n\times(n+k)$ matrix of the form $A=[I_n | B]$ has a $1$-free set of at most $\frac{n}{H_k}+k$ columns where $H_k = \sum_{\ell=1}^k\frac{1}{\ell}$ is the $k$-th harmonic sum.
    For fixed $k$, the bound is tight, i.e., $\frac{n}{H_k}-\exp((1+o_k(1))k) \le u_I(n,n+k)$.
\end{theorem}

\begin{proof}
    The vector $z\in\{0,1\}^{n+k}$ is the characteristic vector of a $1$-free set of columns in $A$ if and only if the vector $Az$ (computed over the reals) has no $1$-coordinates.
    It is convenient to express $z$ as the concatenation of $x\in\{0,1\}^{n}$ and $y\in\{0,1\}^{k}$.
    Clearly, $Az=x+By$, and $Az$ has no $1$-coordinates if and only if $x$ is the characteristic vector of a set that contains every $1$-coordinate in $By$, and none of its $0$-coordinates.
    Therefore, the smallest size of a $1$-free set of columns in $A$ is
    \begin{equation}\label{eq:stan_form}
        u(A)=\min\left\{ \|y\|_1 + \text{~the number of~}1\text{-coordinates in~}By\right\} \text{~over all~} \mathbf{0}\neq y\in\{0,1\}^{k}.
    \end{equation}
    
    In this view, the order at which $B$'s rows appear is immaterial, and all we care is, given $v\in\{0,1\}^k$, how many rows in $B$ equal $v$.
    We denote this number by $c_v$.
    If a row of $B$ equals $v$, then the corresponding coordinate in $By$ equals $\langle v, y\rangle$, where $\langle \cdot, \cdot\rangle$ stands for inner product over the reals.
    Therefore, we can rewrite Equation (\ref{eq:stan_form}) as
    \begin{equation}\label{eq:stan_form_2}
        u(A)=\min\left\{ \|y\|_1 + \sum_{v,\,\langle v,y\rangle=1}c_v\right\}\text{~over all~} \mathbf{0}\neq y\in\{0,1\}^{k}.
    \end{equation}
    
    It transpires that the theorem seeks the maximum of the expression in Equation (\ref{eq:stan_form_2}), and this over the choice of the matrix $B$.
    However, as we saw, all we need to know about $B$ are the nonnegative integers~$c_v$, where clearly $\sum_{v\in\{0,1\}^k}c_v=n$.
    This means that the answer to our problem can be expressed as the optimal value of an integer linear program.
    
    To simplify matters, we neglect the term $\|y\|_1$ in Equation~(\ref{eq:stan_form_2}).
    This causes only a minor loss, since $1\le\|y\|_1\le k$, while $n\to\infty$.
    Furthermore, without loss of generality, to maximize the expression in Equation (\ref{eq:stan_form_2}) we can set $c_{\mathbf{0}}=0$.
    Having made these simplifications, we seek the smallest integer $m=m(k,n)$ for which the following statement holds true:
    \begin{align*}
        \begin{array}{c}
            \text{Given~any~} 2^k-1 \text{~integers~} c_v\ge 0 \text{~indexed by~}v\in\{0,1\}^k\setminus\{\mathbf{0}\} \text{~with~} \sum_{\mathbf{0}\ne v\in\{0,1\}^k}c_v=n,\\[1em]
            \text{~there is some~}\mathbf{0}\neq y\in\{0,1\}^{k} \text{~with~} \sum_{v,\langle v,y\rangle=1}c_v\le m.
        \end{array}
    \end{align*}
     
    Let $M$ be the $(2^k-1)\times(2^k-1)$ binary matrix that is indexed by $\{0,1\}^k\setminus\{\mathbf{0}\}$, whose $(u,v)$ entry equals $1$ if and only if $\langle u,v\rangle=1$.
    The above statement can be restated as follows:
    \begin{align*}
        \begin{array}{c}
            \forall\text{\,integer vector~} c\ge \mathbf{0} \text{~with~} \sum_{\mathbf{0}\ne w\in\{0,1\}^k}c_w=n, 
            \text{~some coordinate in~}Mc\text{~corresponding to a}\\[1em]
            \mathbf{0}\ne y\in\{0,1\}^k\text{~is~}\le m.
        \end{array}
    \end{align*}
    
    In other words, $m=m(k,n)$ is the largest integer $t$ for which there is some $c$ as above such that
    \begin{align*}
        M c\ge \mathbf{1}\cdot t.
    \end{align*}
    It follows that
    \begin{align*}
        \Phi \le u_I(n,n+k)\le \Phi+k,
    \end{align*}
    (the $+k$ slack reflects the neglection of the term $\|y\|_1$) where
    \begin{align*}
        \Phi=\,&\max t, \\
        \text{subject to }~~ & M c\ge \mathbf{1}\cdot t,\numberthis\label{ineq-ilp}\\
        \langle c,\mathbf{1}\rangle=n &\text{ and }  c\ge \mathbf{0} \text{ is a vector of integers},
    \end{align*}
    where $\mathbf{1}\in\{0,1\}^{2^k}$ is the all-$1$ vector.
    
    We turn to solve the rational relaxation of the above ILP.
    \begin{align*}
        &\max t,\\
        \text{subject to }~~ & M c\ge \mathbf{1}\cdot t,\numberthis\label{ineq}\\
        & \langle c,\mathbf{1}\rangle=n \text{ and }  c\ge \mathbf{0}.
    \end{align*}
    
    Clearly, this optimum is an upper bound on $\Phi$.
    Using the idea of LP duality, we left-multiply Inequality~(\ref{ineq}) by the following $2^k$-dimensional vector $w$:
    \begin{align*}
        w_v=\frac{1}{\binom{k-1}{|v|-1}}\text{~for every~}{v\neq\bf0} \text{~in~}\{0,1\}^k.
    \end{align*}
    Clearly, if $v\in\{0,1\}^k$ with $|v|=j$ for some $1\le j\le k$ then
    \begin{align*}
        (w^TM)_v=\sum_{i=1}^k \frac{1}{\binom{k-1}{i-1}}j\binom{k-j}{i-1}=\frac{j!(k-j)!}{(k-1)!}\sum_{i=1}^k \binom{k-i}{j-1}=\frac{j!(k-j)!}{(k-1)!}\binom{k}{j}=k.\numberthis\label{long_calc}
    \end{align*}
    The first equality follows from the definition.
    The second only involves reorganizing terms.
    The third one uses the standard and easy fact that for all positive integers $s \le N$ it holds that
    \begin{align*}
        \sum_{s\le r \le N}\binom{r}{s}=\binom{N+1}{s+1}.
    \end{align*}
    As we left-multiply Inequality (\ref{ineq}) by $w$, Equation~(\ref{long_calc}) yields $w^TM=\mathbf{1}\cdot k$.
    
    Also
    \begin{align*}
        \langle w,\mathbf{1}\rangle = \sum_{i=1}^k\frac{\binom{k}{i}}{\binom{k-1}{i-1}}=kH_k.
    \end{align*}
    Therefore,
    \begin{align*}
        kH_kt=w^T\mathbf{1}\cdot t\le w^TMc\le k\cdot\mathbf{1}^Tc\le kn~\Rightarrow~t\le\frac{n}{H_k}.
    \end{align*}
    The optimum of the LP is an upper bound on the optimum of the ILP.
    Therefore
    \begin{align*}
        u_I(n,n+k)\le\frac{n}{H_k}+k.
    \end{align*}
    
    For the lower bound in the theorem we need a lower bound on the ILP (\ref{ineq-ilp}).
    To this end we define
    \begin{align*}
        z\coloneqq n\frac{w}{kH_k},
    \end{align*}
    and observe that the above calculations yield $\langle z,\mathbf{1}\rangle=n$, and $M z=\mathbf{1}\frac{n}{H_k}$.
    So $\frac{n}{H_k}$ is a lower bound on the LP~(\ref{ineq}).
    We now use this to lower bound the ILP.
    
    Let $H_k\coloneqq\frac{a_k}{b_k}$ written as a reduced rational.
    If $n$ is divisible by $a_k$, then $u_I(n,n+k)=\frac{n}{H_k}+k'$, for some $0\le k'\le k$, because in this case the optimal solutions to our LP and the ILP coincide.
    In the general case, say $n=n_1a_k+n_2$ where $0\le n_2<a_k$, then by the structure of our matrices,
    \begin{align*}
        u_I(n,n+k)\ge u_I(n,n+k)-u_I(n_2,n_2+k)\ge n_1b_k-(n_2+1)\ge n_1b_k-a_k.
    \end{align*}
    
    There holds
    \begin{align*}
        a_k=H_kb_k\le H_k\lcm(1,\ldots,k)&=(\ln k + \gamma +O_k(1/k))\exp((1+o_k(1))k)\\
        &=\exp((1+o_k(1))k),
    \end{align*}
    where we employed the following two well known facts:
    \begin{itemize}
    \item 
    $H_k=\frac{a_k}{b_k}=\ln k + \gamma +O_k(1/k)$, where $\gamma$ is Euler's constant.
    \item 
    $b_k=\lcm(1,\ldots,k)=\exp((1+o_k(1))k)$, where $\ln(\lcm(1,\ldots,k))$ is Chebyshev's second function.
    \end{itemize}
    
    The lower bound on $u_I(n,n+k)$ now follows.
\end{proof}

This implies \cref{main:conj} for these matrices:

\begin{corollary}
    For every $n\ge400$ and $k\ge4$ there holds $u_I(n,n+k)\le0.49n$.
\end{corollary}

\begin{proof}
    When $k\ge4$, $u_I(n,n+k)\le u_I(n,n+4)\le\frac{n}{H_4}+4=0.48n+4\le0.49n$ for all $n\ge400$.
\end{proof}

%%%%%%%%%%%%%%%%%%%%%%%

\section{Some Useful Constructions}\label{sec:construct}

In what follows we make occasional comments that pertain to some well-known families of codes, that the reader may find interesting.
However, no familiarity with coding theory is assumed.
All the relevant background information may be found in standard texts such as \cite{richardson2008modern} and \cite{van1998introduction}.

In this section we introduce several constructions of binary matrices.
The building blocks of these constructions are binary $(2^k-1-k)\times(2^k-1)$ matrices called $U_k$ that we define for $k=2,3,\ldots$.
Further building blocks are $U_{k,m}$, which are $((2^k-1)m-k)\times((2^k-1)m)$ binary matrices.
We define $U_k$ both recursively and directly.
It is easy to verify by a simple inductive argument that the two definitions coincide.
Here is the recursive definition:
\begin{align}
    U_2&=\begin{pmatrix}
    1&1&1
    \end{pmatrix},
\end{align}
\begin{equation}\label{recurse:u}
    U_{k+1}=\begin{pmatrix}
    I_{2^k-1}&\mathbf{1}_{2^k-1\times1}&I_{2^k-1}\\
    \mathbf{0}&\mathbf{0}&U_k
    \end{pmatrix},    
\end{equation}
where $\mathbf{1}_{p\times q}$ is the all-$1$ matrix of dimensions $p\times q$.

In the direct definition of $U_k$ we index its columns by all integers $2^k-1, 2^k-2,\ldots,1$, in this order.
The rows are indexed by the subsequence of the above excluding the powers of $2$.
Each row of $U_k$ has weight $3$.
If the integer $m\in\{1,\ldots,2^k-1\}$ is not a power of $2$, such that $2^t<m<2^{t+1}$ for an integer~$t$, then the three $1$ entries in row $m$ appear in columns $m, m-2^t$ and $2^t$.

For example,
\begin{align*}
    U_3=\begin{pmatrix}
    1&0&0&1&1&0&0\\
    0&1&0&1&0&1&0\\
    0&0&1&1&0&0&1\\
    0&0&0&0&1&1&1
    \end{pmatrix},
\end{align*}
with rows called $7,6,5,3$ in this order and columns called $7,\ldots,1$.

We note that $U_k$ is, by definition, a generator matrix of the $[2^k-1,2^k-1-k,3]_2$ Hamming code, and a parity check matrix of the corresponding simplex code (shortened Hadamard code).
Its rows are linearly independent, since it has an upper-triangular square submatrix with $1$'s on its diagonal.
This matrix, called $T_k$, is attained by erasing those columns of $U_k$ that correspond to powers of two.
The remaining submatrix is called $R_k$, namely the columns of $U_k$ whose index is a power of $2$.
For example,
\begin{align*}
    T_3=\begin{pmatrix}
    1&0&0&1\\
    0&1&0&0\\
    0&0&1&0\\
    0&0&0&1
    \end{pmatrix},~
    R_3=\begin{pmatrix}
    1&0&0\\
    1&1&0\\
    1&0&1\\
    0&1&1
    \end{pmatrix}.
\end{align*}

An interesting note is that this construction exploits the duality between Hamming codes and simplex codes.
A Hamming code is generated by a matrix with row weights $3$, so it is relatively easy to derive a lower bound $b$ on $u(U_k)$.
On the other hand, all non-zero codewords in the order-$N$ simplex code have weight $\frac{N+1}{2}$.
So for $N=2^k-1$ this duality yields an upper bound of $2^{k-1}$ on $\e(U_k)$.
The resulting inequality reads $b\le u(U_k)\le\e(U_k)\le2^{k-1}$, and when $b=2^{k-1}$ holds true, they are both equalities.

We next construct the following $((2^k-1)m-k)\times((2^k-1)m)$ binary matrix for every $k\ge 2, m\ge 1$:
\begin{align}\label{eq:ukm}
    U_{k,m}=\left(\begin{array}{ccc|c|c}
        T_k&&&&R_k\\
        &\ddots&&&\vdots\\
        &&T_k&&R_k\\\hline
        &&&I_{(m-1)k}&\begin{array}{c}I_k\\\hline\vdots\\\hline I_k\end{array}
    \end{array}\right).
\end{align}
The numbers of $T_k$ and $R_k$ blocks are $m$, and empty blocks are all-zero blocks.

%%%%%%%%%%%%%%%%%%%%%%%

\section{\texorpdfstring{\cref{conj:little_oh}}{Conjecture 1.8} Fails for Sub-Logarithmic \texorpdfstring{$\eta$}{eta}}

We prove next that $u(n-\log_2n-1,n-1)=\e(n-\log_2n-1,n-1)=\frac n2$ for infinitely many integers $n$.
Concretely,

\begin{theorem}\label{thm:log_gap}
    For every integer $k\ge 2$ it holds that 
    \begin{align*}
        u(2^k-1-k,2^k-1)=\e(2^k-1-k,2^k-1)=2^{k-1}.
    \end{align*}
\end{theorem}

\begin{proof}
    We proceed by showing that
    \begin{align*}
        \e(2^k-1-k,2^k-1)\le 2^{k-1}\text{~~and~~}u(U_k)\ge 2^{k-1}.
    \end{align*}
    
    We start with the upper bound on $\e$:
    
    \begin{proposition}\label{prop:parity}
        Every $n\times(n+k)$ binary matrix $A$ has an even set of at most $\left(1+\frac{1}{2^k-1}\right)\frac{n+k}{2}$ columns.
        
        In particular, $\e(2^k-1-k,2^k-1)\le 2^{k-1}$.
    \end{proposition}
    \begin{proof}
        We provide two proofs, one linear algebraic and one that appeals to known bounds in coding theory.
        Consider the linear code $\mathcal{C}=\{x\in\F_2^{n+k}\mid Ax=0\}$.
        As usual, we may and do assume that $A$ has rank $n$, so that $|\mathcal{C}|=2^k$.
        Clearly, $\mathcal{C}$ is invariant under addition of any given vector $x\in\mathcal{C}$.
        Consequently, for any $1\le i\le n+k$ either the $i$-th coordinate is identically zero in $\mathcal{C}$ or there are exactly $2^{k-1}$ vectors in $\mathcal{C}$ whose $i$-th coordinate is zero, resp.\ one.
        Consequently, the average Hamming weight of vectors in $\mathcal{C}$ is at most $\frac{n+k}{2}$.
        It follows that the average weight of non-zero codewords in $\mathcal{C}$ is at most $\frac{2^{k-1}(n+k)}{2^k-1}=\left(1+\frac{1}{2^k-1}\right)\frac{n+k}{2}$.
        The conclusion follows.
        
        Alternatively we can appeal to Griesmer's bound \cite{griesmer1960bound} which says that the length of a $k$-dimensional binary linear code of minimum distance $d$ is at least $\sum_{i=0}^{k-1}\ceil{\frac{d}{2^i}}$.
        As we remove ceilings we conclude that this expression is at least $d(2-\frac{1}{2^{k-1}})$.
        After simplifying terms we arrive at $\frac{2^k-1}{2^{k-1}}\e(A)\le n+k$, as claimed.
        
        The second statement follows by letting $n=2^k-1-k$.
    \end{proof}

    We turn to prove that $u(U_k)\ge 2^{k-1}$.
    The proof proceeds by induction on $k\ge2$.
    For $k=2$ the claim clearly holds.
    For the induction step we use the recursive description of $U_{k+1}$ in \cref{sec:construct}.
    Consider a $1$-free set of $U_{k+1}$.
    If it contains column $2^k$, then it must include at least $2^k-1$ additional columns (from either side of the column), for a total of at least $2^k$ columns, as claimed.
    
    Thus it suffices to consider a $1$-free set of columns of the form $L\sqcup R$, where $R,L$ are the subsets of columns from the $2^k-1$ rightmost, leftmost ones, respectively.
    Using the recursive definition in \cref{recurse:u}, we have $|R|\ge 2^{k-1}$ by the induction hypothesis.
    Furthermore, for every $r\in R$, one of its first (upper) $2^k-1$ coordinates contains a $1$.
    Since column $2^k$ is absent from $L\sqcup R$, the (unique) matching column from the $2^k-1$ leftmost columns should be included in $L$, in order for $L\sqcup R$ to be a $1$-free set.
    It follows that $|L\sqcup R|\ge2^k$, completing the proof.
\end{proof}

Notice that the bound in \cref{prop:parity} holds with equality for the parity check matrices of simplex codes.
In this case, the weight of all non-zero codewords coincides with the upper bound in \cref{prop:parity}.

The following extension of \cref{thm:log_gap} yields further cases where $\e$ and $u$ coincide.

\begin{theorem}\label{thm:log-gap generalized}
    For every integers $k\ge2$ and $m\ge1$ it holds that
    \begin{align*}
        u((2^k-1)m-k,(2^k-1)m)=\e((2^k-1)m-k,(2^k-1)m)=2^{k-1}m.
    \end{align*}
\end{theorem}

\begin{proof}
    Again we bound $u$ from below and $\e$ from above.
    The bound on $u$ uses the matrices~$U_{k,m}$ from \cref{sec:construct} and the bound on $\e$ follows from \cref{prop:parity}.
    To show that $u(U_{k,m})\ge 2^{k-1}m$, we consider $1$-free sets in $U_{k,m}$.
    The matrix $U_{k,m}$ with its last $k$ columns removed has no nonempty $1$-free sets, since it is an upper-triangular, full-rank matrix.
    So consider a $1$-free set that includes $t>0$ columns among the last $k$ columns of $U_{k,m}$.
    By \cref{thm:log_gap} at least $(2^{k-1}-t)m$ additional columns are needed, specifically at least $2^{k-1}-t$ from every $T_k$ in the direct sum.
    The lower part of $U_{k,m}$ necessitates adding exactly $(m-1)t$ columns from the columns that contain the block $I_{(m-1)k}$.
    In total the cardinality of the $1$-free set at hand is at least $t+(2^{k-1}-t)m+(m-1)t=2^{k-1}m$, as claimed.
\end{proof}

\cref{thm:log-gap generalized} and the monotonicity of $u,\e$ (see \cref{prop:properties} below) yield

\begin{corollary}\label{prop:weak equivalence corollary}
    For every $k,n$ it holds that $u(n,n+k)\ge\e(n,n+k)-2^{k-1}$.
\end{corollary}

%%%%%%%%%%%%%%%%%%%%%%%

\section{Between \texorpdfstring{$u$}{u} and \texorpdfstring{$\e$}{epsilon} When \texorpdfstring{$n-m$}{n-m} is Bounded}

In this section we compare $u(m,n)$ and $\e(m,n)$ when $n-m\ge1$ is bounded from above.
A~most helpful resource in studying these problems is \cite{grassl2007bounds}, which records best possible linear codes of small lengths.
Below we refer to the notion of an elementary collapse that we now define.
Say that $A$ is a binary matrix and $A(i,j)=1$.
If the $i$-th row has Hamming weight of~$1$, then \emph{elementary collapse} with \emph{pivot} $(i,j)$ is the matrix that is obtained upon removing row $i$ and column $j$ from $A$.

Here is our main result: 

\begin{theorem}\label{thm:sec5 thm}
    \begin{enumerate}
        \item\label{case:e>u}
        $u(4,8)=3$ whereas $\e(4,8)=4$.
        This is the lexicographically first case where $u<\e$, i.e., the ordering that compares $n-m,m,n$ in this order.
        \item\label{case:k=1}
        $u(n,n+1)=\e(n,n+1)=n+1$.
        The case of equality is fully characterized.
        \item\label{case:k=2}
        $u(n,n+2)=\e(n,n+2)=\floor{\frac{2n+4}{3}}$.
        \item\label{case:k=3}   
        If $n\not\equiv -1 \bmod 7$, then $u(n,n+3)=\e(n,n+3)=\floor{\frac{4n+12}{7}}$.
        Also, $u(7m-1, 7m+2)=4m$ for every positive integer $m$.
    \end{enumerate}
\end{theorem}

% TODO: is it better to have all the following in this proof env or not?
\begin{proof}
    We start with several simple observations:
    
    \begin{proposition}\label{prop:properties}
        \begin{enumerate}
            \item
            $u(m,n)\le\e(m,n)\le m+1$.
            \item
            Both $u(m,n)$ and $\e(m,n)$ increase with $m$ and decrease with $n$.
            \item\label{monotonicity}
            $u(m,n)\le u(m+1,n+1),~\e(m,n)\le\e(m+1,n+1)$.
            \item\label{collapse}
            The functions $u(\cdot), \e(\cdot)$ are invariant under elementary collapses.
        \end{enumerate}
    \end{proposition}
    
    \begin{proof}
        \begin{enumerate}
            \item
            $u\le\e$ by definition.
            $\e(m,n)\le m+1$ because the columns of every $m\times (m+1)$ matrix are linearly dependent.
            \item
            This is because every $1$-free/even set of columns in an $(m+1)\times n$ matrix $A$ is $1$-free/even for every $m\times n$ submatrix of $A$, and these properties are preserved when any column is added to $A$.
            \item
            Observe that $u(A)=u(B),~\e(A)=\e(B)$ for every $m\times n$ matrix $B$ and the matrix
            \begin{align*}
                A=\left(\begin{array}{cc}
                    B&\bf0\\
                    \bf0&1
                \end{array}\right).
            \end{align*}
            \item
            If $A$ has an elementary collapse with pivot $(i,j)$, then no even resp.\ $1$-free set of columns of $A$ can contain $j$.
            Therefore $u,\e$ remain unchanged under an elementary collapse with pivot $(i,j)$.
        \end{enumerate}
    \end{proof}

    We now prove the 4 items of the theorem in order.

    \subsection{Proof of \texorpdfstring{\cref{case:e>u}}{Item 1}: \texorpdfstring{$u(4,8)=3<4=\e(4,8)$.}{u(4,8)=3<4=epsilon(4,8)}}
    
    As recorded in \cite{grassl2007bounds}, there holds $\e(4,8)=4$ (see reference for a linear code that attains this bound).
    We now show that $u(4,8)=3$.
    The following matrix yields $u(4,8)\ge3$, since every nonempty set of $2$ or fewer columns has a row of weight $1$.
    \begin{align}
        \begin{pmatrix}
        1&0&0&0&1&1&0&0\\
        0&1&0&0&0&1&1&0\\
        0&0&1&0&0&0&1&1\\
        0&0&0&1&1&0&0&1
        \end{pmatrix}.
    \end{align}
    
    Next we show that $u(A)\le 3$ for every binary $4\times8$ matrix $A$.
    We reduce to the case that every column of $A$ has weight at least $2$.
    If $A$ has a zero column, then clearly $u(A)=1$.
    If some column of $A$ has weight $1$, say $a_{1,1}=1$ and $a_{i,1}=0$ for $i=2,3,4$, consider the submatrix~$B$ of $A$ that is obtained by erasing its first row and column.
    If $B$ has an all-zero column, then $u(A)\le2$, and if $B$ has two equal columns, then $u(A)\le3$.
    In the only remaining case
    \begin{align}
        B=\begin{pmatrix}
        0&0&0&1&1&1&1\\
        0&1&1&0&0&1&1\\
        1&0&1&0&1&0&1
        \end{pmatrix},
    \end{align}
    up to permutations of the rows and columns.
    We index the columns of $B$ with $2,\ldots,8$.
    Consider the weight $w=\sum_j a_{1,j}$ of row $1$ in $A$.
    If $w=1$, then $a_{1,j}=0$ for all $j\ge 2$.
    Consequently, $u(A)=u(B)=3$, since every $1$-free set in $B$ is also $1$-free in $A$.
    If $w\ge 3$ there are at least two indices $1<\beta<\alpha$ such that $a_{1,\alpha}=a_{1,\beta}=1$.
    If columns $\alpha, \beta$ of $B$ are the vectors $u$ resp.\ $v$, then some column $\gamma$ corresponds to $u\oplus v$ ($\bmod~2$ sum).
    Columns $\alpha, \beta, \gamma$ form a $1$-free set in~$A$.
    Finally, if $w=2$, there is exactly one index $\delta>1$ such that $a_{1,\delta}=1$.
    Then we can find a triplet of columns of the form $u, v, u\oplus v$ in $B$ none of which is column $\delta$.
    
    We can now assume that every column of $A$ has weight $2,3$ or $4$.
    $A$ has no repeated columns, or else $u(A)=2$.
    Also $A$ can have at most two columns of weight at least $3$, for any three such distinct vectors form a $1$-free set.
    Consequently $A$ has exactly two columns of weight at least~$3$ and each of the six columns of weight $2$.
    The latter $6$-tuple contains a $1$-free set of three columns.
    This establishes \cref{case:e>u} of \cref{thm:sec5 thm}.
    
    Note that $4,8$ are the \emph{minimal} $m,n$ for which $u(m,n)<\e(m,n)$, with respect to the lexicographical ordering that first compares $n-m$, then $m,n$: 
    the other parts of the present theorem show that equality holds when $n-m\le3$.
    One can verify that $u(m,n)=\e(m,n)$ for $(m,n)\in\{(1,5),(2,6),(3,7)\}$ using the table in \cite{grassl2007bounds}.

    \subsection{Proof of \texorpdfstring{\cref{case:k=1}}{Item 2}}
    
    It is clear that $u(n,n+1)=\e(n,n+1)=n+1$.
    Also $\e(A)=n+1$ for an $n\times(n+1)$ matrix~$A$ if and only if its rank is $n$ and all its row weights are even.
    
    An example of an $n\times(n+1)$ matrix $A$ with $u(A)=n+1$ is obtained by taking the edges vs.\ vertices incidence matrix $A_T$ of a tree with $n+1$ vertices.
    As we show, no other examples exist.
    
    \begin{proposition}\label{prop:tree characterization}
        If $u(A)=n+1$ for some $n\times (n+1)$ binary matrix $A$, then $A=A_T$ for some tree $T$ with $n+1$ vertices.
    \end{proposition}
    
    \begin{proof}
        $A$ cannot have a zero row, or else $u(A)\le\e(A)\le \frac{2(n+1)}{3}$, by \cref{prop:parity}.
        
        As in \cref{prop:properties}, \cref{collapse}, any row of weight $1$ in $A$ can be collapsed, without changing $\e$ and $u$.
        So we may and will assume that every row of $A$ weighs at least $2$.
        Let us view $A$ as the edges vs.\ vertices incidence matrix of a hypergraph $G=(V,E)$.
        An edge in $E$ of size $2$ (resp.\ $\ge 3$) is called \emph{light} (resp.\ \emph{heavy}).
        Let $L\subseteq E$ be the set of light edges.
        If all edges in $E$ are heavy, we can omit any single column of $A$ and obtain a matrix in which all rows weigh at least~$2$, contrary to our assumption that $u(A)=n+1$.
        
        The graph $(V,L)$ has no isolated vertices, for if $v\in V$ is incident with no light edge, then $V\setminus \{v\}$ is a $1$-free set, contrary to our assumption.
        If $L=E$, then the vertex set of any connected component of $G$ is a $1$-free set.
        Therefore $G$ is a connected graph with $n+1$ vertices and $n$ edges, i.e., a tree, as claimed.
        On the other hand, if $L\neq E$, the graph $(V,L)$ must be disconnected, since it has $n+1$ vertices and at most $n-1$ edges.
        In this case, let $(V_1,L_1),\ldots,(V_k,L_k)$ be the connected components of $(V,L)$.
        By the above $\sum_i|V_i|=n+1,~|L_i|\ge|V_i|-1$, so that $|L|=\sum_i|L_i|\ge n+1-k$, with equality if and only if $(V,L)$ is a forest with no isolated vertices.
        Consequently, at most $k-1$ edges in $E$ are heavy.
        
        Let $B$ be the edges vs.\ vertices matrix of the hypergraph that results from $G$ by shrinking each $V_i$ to a single new vertex $v_i$.
        Since $L\neq\emptyset$ this actually reduces the size of the matrix and we can use induction to prove the proposition.
        Every $1$-free set $S$ in $B$ yields a $1$-free set in $A$ by inflating each $v_i\in S$ to $V_i$.
        In particular $u(B)<k$ would imply $u(A)\le n$.
        Consequently, $B$ is a $(k-1)\times k$ matrix with $u(B)=k$.
        By induction it is the edge-vertex matrix of $K$, a tree with vertex set $\{v_1,\ldots,v_k\}$.
        Say that $v_1$ is a leaf of $K$, and let $e$ be the single edge of $K$ that is incident with $v_1$.
        We claim that either $V_1$ or $V\setminus V_1$ comprise a $1$-free set in $A$.
        Indeed, only the row corresponding to $e$ may have weight $1$ in the submatrix of $A$ corresponding to either $V_1$ or $V\setminus V_1$.
        But it is impossible that both cases occur, for that would mean that the edge $e$ has size $2$ contrary to the fact that $e$ is a heavy edge.
        The set of vertices (or columns) $V_1$ or $V\setminus V_1$ in which $e$ has weight at least $2$ comprises a $1$-free set of columns in $A$.
        Since both $V_1, V\setminus V_1$ are nonempty, this contradicts our assumption that $u(A)=n+1$.
        This establishes the proposition and \cref{case:k=1} of \cref{thm:sec5 thm}.
    \end{proof}

    \subsection{Proof of Items \texorpdfstring{\ref{case:k=2}}{3} and \texorpdfstring{\ref{case:k=3}}{4}}
    
    The proof for $k=2$ splits to cases according to the value of $n\bmod 3$.
    When $n\equiv 1\bmod 3$ we have $u(3m-2,3m)=\e(3m-2,3m)=2m$ by \cref{thm:log-gap generalized}.
    By \cref{prop:parity}, $u,\e$ do not change as we move to $n=3m-1$.
    Finally, for $n=3m$ we introduce the matrix
    \begin{align*}
        A\coloneqq\left(\begin{array}{cc}
             U_{2,m}&\mathbf{0}\\
             \begin{array}{cc}\mathbf{0}&I_2\end{array}&I_2
        \end{array}\right),
    \end{align*}
    with $U_{2,m}$ as defined in \cref{eq:ukm}.
    It is easy to see that $u(A)=\e(A)=2m+1$.
    By \cref{prop:parity} this is also the upper bound on $u(3m,3m+2),\e(3m,3m+2)$.
    We conclude that $u(n,n+2)=\e(n,n+2)=\floor{\frac{2n+4}{3}}$, establishing \cref{case:k=2}.
    
    The analysis when $k=3$ is somewhat more involved and proceeds according to the value of $n\bmod 7$.
    We start with the upper bound:
    By \cref{prop:parity}, $\e(n,n+3)\le\floor{\frac{4n+12}{7}}$.
    This bound is shown to be tight, except if $n\equiv -1 \bmod 7$, when it can be reduced by $1$ due to Griesmer's bound \cite{griesmer1960bound}:
    indeed, our general upper bound is $\e(7m-1,7m+2)\le\floor{\frac{28m+8}{7}}=4m+1$, but by Griesmer's bound if the code's distance is $4m+1$, then its length is at least $4m+1+\ceil{\frac{4m+1}{2}}+\ceil{\frac{4m+1}{4}}=7m+3$.
    So $\e(7m-1,7m+2)\le4m$.
    
    We proceed to deal with the lower bounds.
    The case $k=3$ of \cref{thm:log-gap generalized} gives $u(7m-3,7m)=\e(7m-3,7m)=4m$.
    Namely, $u=\e$ when $n\equiv 4\bmod 7$.
    
    \cref{monotonicity} of \cref{prop:properties} and \cref{prop:parity} yield
    \begin{align*}
        u(n-1,n+2)\le u(n,n+3)\le\e(n,n+3)\le\floor{\frac{4}{7}(n+3)}.
    \end{align*}
    So if $u(n-1,n+2)=\floor{\frac{4}{7}(n+3)}$, this trivially allows to derive the case $n\equiv r+1\bmod 7$ after establishing $u=\e$ in the case $n\equiv r\bmod 7$.
    This works verbatim for $r=1,4$.
    When $n\equiv5\bmod7$, a similar argument establishes $u(7m-1,7m+2)=\e(7m-1,7m+2)=4m$.
    It is left to establish the cases $n\equiv 0,1,3\bmod7$.
    Here, an additional argument is needed.
    To this end, we extend $U_{3,m}$ from \cref{sec:construct} to an $n\times (n+3)$ matrix for the appropriate $n$.
    This resembles the construction of $U_{k,m}$ from $U_k$, and the construction in the case $k=2$.
    In all three cases, these matrices show that $u(n,n+3)$ attains the upper bound on $\e(n,n+3)$, namely~$\floor{\frac{4}{7}(n+3)}$.
    Hence we get in each case a matrix $U$ such that $\e(n,n+3)\le\floor{\frac{4}{7}(n+3)}\le u(U)\le u(n,n+3)$.
    For illustration, when $n=7m$, we use the matrix $U_{3,m}$ to construct
    \begin{align*}
        U\coloneqq\left(\begin{array}{cc}
             U_{3,m}&\mathbf{0}\\
             \begin{array}{cc}\mathbf{0}&I_3\end{array}&I_3
        \end{array}\right).
    \end{align*}
    Note that $4m+1=u(U)\le u(7m,7m+3)$ and $\e(7m,7m+3)\le 4m+1$ from \cref{prop:parity}, so $u(7m,7m+3)=\e(7m,7m+3)=4m+1$.
    
    We note that \cref{case:k=3} holds as well when $n=1,2,3$, but we skip this verification.
    \bigbreak
    This concludes the proof of \cref{thm:sec5 thm}.
\end{proof}

%%%%%%%%%%%%%%%%%%%%%%%

\section{Open Problems}

\begin{problem}
    The most obvious question is \cref{main:conj} which remains open.
\end{problem}

\begin{problem}
    What is the smallest $c$ for which the conclusion of \cref{thm:min-weight rows} holds?
    Is it $5$?
\end{problem}

\begin{problem}
    The proof of \cref{thm:min-weight rows} suggests a more general setup.
    We seek a $1$-free set of columns in a binary matrix $A$.
    Having committed to some subset of columns, the rows of $A$ are split into:
    $I_0\sqcup I_1\sqcup I_*$, those of weight $0, 1$ and $\ge 2$, respectively.
    To extend our initially chosen set into a $1$-free set, we need an additional set of columns $J$, the weight of whose $I_0$ and $I_1$ rows differ from $1,0$ respectively.
    Under what conditions is it possible to pre-specify which row sums we wish to be $\neq 0$ and which $\neq 1$?
\end{problem}

\begin{problem}
    Let $u_3(m,n)$ denote $\max u(A)$ of an $m\times n$ binary matrix $A$ where every row has weight $3$.
    \cref{prop:tree characterization} implies that $u_3(n,n+1)<u(n,n+1)$, but perhaps $u_3(m,n)=u(m,n)$ when $n+1<m$.
    Some supportive evidence for this is that $u_3(4,8)=u(4,8)$, ${u_3(2^k-1-k,2^k-1)=u(2^k-1-k,2^k-1)}$.
    We note that more generally, $u_3((2^k-1)m-1,(2^k-1)m)=u((2^k-1)m-1,(2^k-1)m)$ holds, because the matrices $U_{k,m}$ can be modified so all rows have weight $3$ without changing $u,\e$.
\end{problem}

\begin{remark}
    If indeed, \cref{main:conj} is true, then its proof would require some substantial new ideas.
    Because, as \cref{thm:log_gap} shows, methods that work for square matrices and matrices with only a few more columns than rows as in \cref{thm:min-weight rows} and \cref{thm:log-gap generalized} are not likely to deliver a full answer.
\end{remark}

\bibliographystyle{alphaurl}
\bibliography{bibliography}

\end{document}